\documentclass{article}
\usepackage[utf8]{inputenc}
\usepackage{fullpage}
\usepackage{algorithm} 
\usepackage{algpseudocode} 
\usepackage{amsthm}
\usepackage{amsmath}
\usepackage{hyperref}

\title{An Algorithm for Consensus Trees}

\author{Pongsaphol Pongsawakul \\ \href{mailto:pongsaphol@pongsaphol.com}{pongsaphol@pongsaphol.com}}

\newtheorem{lemma}{Lemma}
\newtheorem{claim}{Claim}
\newtheorem{theorem}{Theorem}

\begin{document}
 
\maketitle

\begin{abstract}
We consider the tree consensus problem, an important problem in bioinformatics.   Given a rooted tree $t$ and another tree $T$, one would like to incorporate compatible information from $T$ to $t$.  This problem is a subproblem in the tree refinement problem called the RF-Optimal Tree Refinement Problem defined by in Christensen, Molloy, Vachaspati and Warnow [WABI'19] who employ the greedy algorithm by Gawrychowski, Landau, Sung, and Weimann [ICALP'18] that runs in time $O(n^{1.5}\log n)$.  We give a faster algorithm for this problem that runs in time $O(n\log n)$.  Our key ingredient is a bipartition compatibility criteria based on amortized-time leaf counters.  While this is an improvement, the fastest solution is an algorithm by Jansson, Shen, and Sung [JACM'16] which runs in time $O(n)$.
\end{abstract}

\section{Introduction}

We consider the tree consensus problem, an important problem in bioinformatics.  Given a rooted tree $t$ and another rooted tree $T$, we would like to combine ``information'' from $T$ into $t$.  More over, we would like to only greedily take information that is currently consistent with our current $t$.  (See definitions below.)
This problem is a subproblem in the tree refinement problem called RF-Optimal Tree Refinement Problem defined by in Christensen, Molloy, Vachaspati and Warnow~\cite{ChristensenMVW-traction19} who employ the greedy algorithm by Gawrychowski, Landau, Sung, and Weimann~\cite{GawrychowskiLSW18} that runs in time $O(n^{1.5}\log n)$.  We give a faster algorithm for this problem that runs in time $O(n\log n)$.  Our key ingredient is a bipartition compatibility criteria based on amortized-time leaf counters.  While this is an improvement, the fastest solution is an algorithm by Jansson, Shen, and Sung~\cite{JanssonSS16} which runs in time $O(n)$.

The algorithm by Gawrychowski {\em et al}~\cite{GawrychowskiLSW18} works in a more general case where the goal is the find the greedy consensus trees from $k$ trees.  In this case, their algorithm runs in time $O(kn^{1.5}\log n)$, an improvement over $O(kn^2)$ of Jansson {\em et al}~\cite{JanssonSS16}.  For this problem, Sung~\cite{Sung18} also present an algorithm that runs in time $O(k^2 n)$, improving over Gawrychowski {\em et al}~\cite{GawrychowskiLSW18} when $k=O(\sqrt{n}\log n)$.

In an earlier version, we erroneously claimed that our algorithm works for the case with many trees.  We thank Pawel Gawrychowski and Oren Weiman for pointing this out.  Jittat Fakcharoenphol who help advising the author on this manuscript would like to take the full responsibility for this mistake.

\section{Definitions}

We start with definitions related to trees and consistency.  

Let $V(T), E(T)$ and $R(T)$ be vertex set, edge set and the root of tree $T$. 
For every vertex $u \in V(T)-\{R(T)\}$, let $par(u)$ be parent node of node $u$.
For every vertex $u \in V(T)$, let $depth(u)$ be depth of node $u$. We can denote as $depth(u) = depth(par(u)) + 1$, for every vertex $u \in V(T)-\{R(T)\}$ and $depth(R(T)) = 1$.
For every vertex $u \in V(T)$, let $L(u)$ be set of all leaves on subtree $u$. 
Let $size(u) = |L(u)|$ for each $u \in V(T)$.
For each node $u \in V(T)-\{R(T)\}$, we call $\Lambda(u)$ be set of {\em bipartition} at edge $(u, par(u))$. 
For each bipartition, $\Lambda(u)$, we can represent into two clusters, $A = L(u)$ and $B = L(R(T))-L(u)$, and denoted by $A|B$.
The set of bipartitions of $T$ can denoted by $C(T) = \{\Lambda(u) : u \in V(T)-\{R(T)\} \}$.
Let $RF(T_a, T_b)$ be {\em Robinson-Foulds} distance between trees $T_a$ and $T_b$. RF-distance can denoted by $RF(T_a, T_b) = |C(T_a)-C(T_b)|$.

The set $S$ of bipartitions is compatible if there exists tree $T$ such that $C(T) = S$.
leaves set $A$ is compatible with $t'$ when have node $u$ that for all $v \in child(u)$ if only if $L(v) \cap S = \emptyset$ or $L(V) \subseteq S$.

%

\section{The \texorpdfstring{$O(n^2)$}{} algorithm}

In this section, we describe a simpler version of the algorithm and prove its correctness and its running time (in Subsection~\ref{sect:basic-running-time}).  
We improve its running time in Section~\ref{sect:faster-alg}

We assume that both $T$ and the current $t'$ are equipped with a data structure that given an id of a bipartition $b$, find vertex $u$ in the trees such that $\Lambda(u)=b$.  Since $t'$ changes over time, we assume that our data structures can handle the tree update efficiently.  
We discuss this in Subsection~\ref{sect:update}.

The main loop of the algorithm iterates over all bipartition of $T$ recursively and, if possible, add each bipartition to $t'$. 

We define variables used in the main loop. 
Let $z$ be the node that have minimum depth in tree $t'$ for each bipartition in $T$.
The main loop is described below.

\begin{algorithm}[H]
	\caption{Main loop}
	\begin{algorithmic}[1]
	    \For{each node $u \in V(T)$}
	        \State {\sc ClearCounter()} 
	        \State $z\leftarrow$ {\sc UpdateCounterSubtree}($u$)  \Comment{Update Step, referred to in Section~\ref{sect:faster-alg}}
	        \If{{\sc IsCompatible}($t, u, z$)}
	            \State update $t'$
	        \EndIf
	    \EndFor
	\end{algorithmic}
\end{algorithm}

The main loop uses the following function.

\begin{algorithm}[H]
	\caption{{\sc UpdateCounterSubtree}($u$)}
	\begin{algorithmic}[1]
        \State $z\leftarrow null$     
        \For{each node $v \in L(u)$}
            \State $p \leftarrow$ {\sc UpdateCounterLeaf}($v,t'$) \Comment{After this step we say that $v$ has been {\em added}.}
            \If{$z = null$ or $depth(p) < z$}
                \State $z \leftarrow p$
            \EndIf
        \EndFor
        \State return $z$
	\end{algorithmic}
\end{algorithm}

Our algorithm maintains variable $counter$ for each vertex in $t'$.  We also keep a list of dirty vertices so that {\sc ClearCounter} can run in $O(1)$ time.  Variable $counter$ is updated in function {\sc UpdateCounterLeaf}.  Note that $counter$ changes over time as {\sc UpdateCounterSubtree} keeps adding leaves in $L(u)$.   The algorithm ensures that $counter(u)$ is exactly as follows.  If $u$ is a leaf vertex, we let
\[
 counter(u) = 
  \begin{cases}
   1 & \text{if } v \text{ has been added} \\
   0 & \text{otherwise} 
  \end{cases} 
\]

For other internal vertex $u$, we let
\[
 counter(u) = 
  \sum_{v \in child(u)}
  \begin{cases} 
    counter(v) & \text{if } counter(v) = |L(v)| \\
    0      & \text{otherwise} 
  \end{cases}
\]

For each call this function, it can take amortized $O(1)$ time (to be proved later). The following algorithm describes function {\sc UpdateCounterLeaf}.

\begin{algorithm}[H]
	\caption{{\sc UpdateCounterLeaf}($v,t'$)}
	\begin{algorithmic}[1]
	    \State $counter(v) \leftarrow 1$
	    \While{$counter(v)=size(v)$}
	        \State $p \leftarrow par(v)$
	        \State $counter(p) \leftarrow counter(p) + counter(v)$
	        \State $v \leftarrow p$
	    \EndWhile
	    \State return $v$
	\end{algorithmic}
\end{algorithm}

Next, we have algorithm that check that add from previous algorithm is compatible to $t'$. 
We can track the node $u$ that have minimum depth and have $counter(u) > 0$ in the while loop at the previous algorithm. 

\begin{algorithm}[H]
    \caption{{\sc IsCompatible}($t', u, z$)}
    \begin{algorithmic}[1]
        \If{$counter(z) = |L(u)|$}
            \State return YES
        \Else
            \State return NO
        \EndIf
    \end{algorithmic}
\end{algorithm}

We prove the correctness of the algorithm.  We note that it considers all bipartitions.

\begin{lemma}
The main loop considers all bipartitions in $T$.
\end{lemma}
\begin{proof}
From definition, $\Lambda(u)$ can represent all bipartitions of $T$ and for each bipartition, we have consider all node in $\Lambda(u)$. From this, we can consider all bipartitions in $T$. 
\end{proof}

For clarity, for each vertex $w\in t'$ we denote by $L_{t'}(w)$ its leaf set in $t'$.
Note each call to {\sc UpdateCounterLeaf} increases $counter$ for each vertex by at most 1, and we call {\sc UpdateCounterLeaf} exactly $|L(u)|$ times.  This implies the next lemma, which can be formally proven by induction.  

\begin{claim}
During the call of {\sc UpdateCounterSubtree($u$)}, for any vertex $v\in t'$, $counter(v)\leq |L(u)\cap L_{t'}(v)|$.  Moreover, if $L_{t'}(v)\subseteq L(u)$, $counter(v)=|L(u)\cap L_{t'}(v)|= |L_{t'}(v)|$, i.e., the counter attains its maximum.
\label{lem:upperboundcounter}
\end{claim}

The following the key lemma.

\begin{lemma}
   Let $z'$ be the least common ancestor of leaf vertices in $L(u)$ in $t'$.
   After {\sc UpdateCounterSubtree($u$)} is called, leaf set $L(u)$ is compatible with the current $t'$ if and only if $counter(z')=|L(u)|$.
\end{lemma}
\begin{proof}
Note that after the call to {\sc UpdateCounterSubtree($u$)}, we have called {\sc UpdateCounterLeaf($v$)} for every leaf $v$ in $L(u)$.  From Claim~\ref{lem:upperboundcounter}, we only need to consider the case when $counter(z')\leq |L(u)|$.

The algorithm maintains vertex $z$, which is the vertex closest to the root that {\sc UpdateCounterLeaf} has touched.  
We first consider the case that $z=z'$.

We show that if $counter(z')=|L(u)|$, then for each child $w\in children(z')$, $L_{t'}(w)\cap L(u)=\emptyset$ or $L_{t'}(w)\subseteq L(u)$.  This implies that $L(u)$ is compatible with $t'$.

We note that $z'$ is not a leaf.   Consider each $w\in children(z')$.  We only need to consider $w$ such that $L_{t}(w)\cap L(u)\neq\emptyset$.  The only way $counter(z')=|L(u)|$ is when $w$ is ``complete'', i.e., $counter(w) = |L_{t'}(w)|$.  Since $counter(w)\leq |L(u)\cap L_{t'}(w)|$ (from Claim~\ref{lem:upperboundcounter}), we know that $L_{t'}(w)\subseteq L(u)$.

On the other hand, if $L(u)$ is compatible with $t'$, we show that $counter(z')=|L(u)|$.  We prove a stronger statement: if $L(u)$ is compatible with $t'$ for every vertex $v\in t'$ in the subtree rooted at $z'$, 
\[
counter(v)=|L(u)\cap L_{t'}(v)|,
\]
i.e., the upper bound in Claim~\ref{lem:upperboundcounter} attains its maximum.  To do so, we prove inductively on the structure of $t'$.  Clearly, the claim is true when $v$ is a leaf.  Consider vertex $v\neq z'$ in the subtree of $t'$ rooted at $z'$.  If $L_{t'}(v)\cap L(u)=\emptyset$, $counter(v)=0$; thus the property follows.  Now, consider $v$ such that $L_{t'}(v)\cap L(u)\neq \emptyset$.  Let $w$ be a child of $z'$ such that $v$ belongs to subtree rooted at $w$.  Since $L_{t'}(v)\cap L(u)\neq\emptyset$ and $L_{t'}(v)\subseteq L_{t'}(w)$, we know that $L_{t'}(w)\cap L(u)\neq\emptyset$.  Since $L(u)$ is compatible with $t'$, we have that
\[
L_{t'}(w)\subseteq L(u),
\]
implying that $L_{t'}(v)\subseteq L(u)$; thus $counter(v)=|L_{t'}(v)|=|L_{t'}(v)\cap L(u)|$, from Claim~\ref{lem:upperboundcounter}.

Finally, consider $z'$.  Note that since $z'$ is the common ancestor of leaves in $L(u)$, $L_{t'}(z')\supseteq L(u)$.  For each child $w$ of $z'$, when $L_{t'}(w)\cap L(u)\neq\emptyset$, $w$ is complete and propagate $|L_{t'}(w)\cap L(u)|$ to $counter(z')$.  Summing all children of $z'$, we have that $counter(z')=|L(u)|$.

This completes the proof of the lemma.

\end{proof}

\begin{lemma}
Tree compatibility condition works
\end{lemma}

From above condition, our algorithm have $add$ function $counter(u)$. for each subtree $L(u)$ have fully resolved when $counter(u) = |L(u)|$

\subsection{Running time analysis}
\label{sect:basic-running-time}

We first analyze the running time of the algorithm except the calls to {\sc UpdateCounterSubtree}.  We show that this part runs in linear time.  

We start with {\sc UpdateCounterLeaf}.

\begin{lemma}
Function {\sc UpdateCounterLeaf}($v, t'$) runs in amortized $O(1)$ time.
\label{lemma:counter-leaf}
\end{lemma}
\begin{proof}
We use the potential method.  Our data structure consists of variables $counter$ for all vertices in $t'$.   Denote the data structure at time $i$ by $D_i$.  We say that a vertex $u\in t'$ is {\em incomplete} if $0 < counter(u) < |L(u)|$.
Let potential function $\Phi(D_i)$ be the number of incomplete vertices in $t'$ time $i$.  
Using the potential method, when the data structure changes from $D_{i-1}$ to $D_i$, the amortized cost of an operation is $\hat{c} = c + \Delta\Phi$, where $c$ is an actual cost, and $\Delta\Phi = \Phi(D_i)-\Phi(D_{i-1})$.   Let $D_0$ be initial data structure after {\sc ClearCounter} is called; thus $\Phi(D_0)=0$.  Note that $\Phi(D_i)\geq \Phi(D_0)=0$ for any $i$.

When invoking {\sc UpdateCounterLeaf} at time $i$, let $k$ be number of times the while loop in Lines 2 - 6 is executed.  Clearly, the actual cost $c$ of the operation is $k+1$.  Let $\Delta\Phi = \Phi(D_i)-\Phi(D_{i-1})$. 

We claim that $\Phi(D_i) = \Phi(D_{i-1})-k+1$, i.e., the number of incomplete vertices decreases by $k-1$.  Let $v'$ be the actual leaf that {\sc UpdateCounterLeaf} is called on.  Note that each time the loop is executed, $counter(v) = size(v) = |L(v)|$.  Except when $v=v'$, previously at time $i-1$, we know that $0 < counter(v) < |L(v)|$, because $v$ is an internal vertex with at least 2 children; hence $v$ was incomplete at time $i-1$.  Since $counter(v)=|L(v)|$ at time $i$, $v$ is no longer incomplete; thus the number of incomplete vertices decreases by $k-1$ as claimed.

Thus the amortized cost $\hat{c}=k+1 + \Delta\Phi = k+1 + (-k+1)=2=O(1)$.

\end{proof}

We now analyze the running time of {\sc UpdateCounterSubtree}.

\begin{lemma}
Function {\sc UpdateCounterSubtree}($u$) runs in time $O(|L(u)|)$.
\label{lem:subtree-runtime}
\end{lemma}
\begin{proof}
From Lemma~\ref{lemma:counter-leaf}, it is clear that {\sc UpdateCounterSubtree} runs in time $O(|L(u)|)=O(n)$ and it is invoked for $O(n)$ time.  Therefore, the total running time of the function is $O(n^2)$.  Combining the two parts, we get that the algorithm runs in $O(n^2)$.
\end{proof}

\subsection{Updating \texorpdfstring{$t$}{}}
\label{sect:update}

In this section, we show that when the bipartition defined by $L(u)$ is compatible with $t'$, we can update $t'$ to include that bipartition efficiently in time $deg(u)$.  When $t'$ is compatible with a bipartition defined by $L(u)$ for $u\in T$, to update $t'$ we have to create a new child of $z$ that consists of only children of $z$ that corresponds to the bipartition $L(u)$.  Note that these children are those ``full'' $counter$ that also propagate the counter to $z$.  Therefore, when a child propagate a counter to any vertex, we keep a list of them.  When we need to update $t'$ at $z$, we can take every vertex in this list, and create a new child $z'$ of $z$ with these vertices as $z'$'s children and update their counter accordingly.  This can be done in time $O(deg(u))$.


\section{The faster algorithm: heavy child optimization}
\label{sect:faster-alg}

In this section, we describe a simple method to speed up the algorithm from the previous section.  Note that the only bottle-neck to a nearly linear time algorithm is the counting procedure.

As a preprocessing, we assume that for each vertex $u\in T$, we know $|L(u)|$.  This can be computed in $O(n)$ time.  The improved algorithm is described below.

\begin{algorithm}[H]
	\caption{{\sc Solve}($u$)} 
	\begin{algorithmic}
	    \If{$u$ is not leaf}
	        \State let $c_1,c_2$ be children of $u$  \Comment{There are exactly two children, since $T$ is binary}
	        \If{$|L(c_1)| > |L(c_2)|$}
	            \State swap node $c_1$ and node $c_2$
	        \EndIf
	        \State {\sc Solve}($c_1$) \Comment{Solve smaller subtree}
	        \State {\sc ClearCounter()}
	        \State {\sc Solve}($c_2$) \Comment{Solve larger subtree}
	    \EndIf
	    \State $z \leftarrow$ {\sc UpdateCounterSubtree}($c_1$) \Comment{Update the counter for the smaller subtree}
        \If{{\sc IsCompatible}($t, u, z$)}
            \State update $t'$
        \EndIf
	\end{algorithmic} 
\end{algorithm}

To see that this algorithm is the correct implementation of the Main loop, we essentially need to show that at the end of {\sc Solve}($u$), variable $counter$ is exactly equal to variable $counter$ right after the ``Update Step'' in Line 3 in the Main loop while processing $u$, i.e., variable $counter$ is exactly equal to the case when every leaf in $L(u)$ has been added while no other leaves have been added.   This can be shown by induction on the calls of {\sc Solve}.  We omit the proof in this version of the manuscript.

We are left to analyze its running time.

\begin{theorem}
The algorithm {\sc Solve} runs in $O(n \log n)$ time.
\end{theorem}

\begin{proof}
Note that the running time for all other operations in {\sc Solve} is $O(1)$ per invocation.  Since {\sc Solve} is called for $O(n)$ time, the total running time of these operations is $O(n)$.  Also, the running time of {\sc ClearCounter()} can be amortized to the running time of {\sc UpdateCounterSubtree}, where the counters are updated.

Therefore, we are left to analyze the running time of {\sc UpdateCounterSubtree}.

Note that {\sc UpdateCounterSubtree}($u$) for $u\in T$ runs in time linearly in the number of leaves, $|L(u)|$, from Lemma~\ref{lem:subtree-runtime}.  Hence, we can charge the cost to these leaves.  

We analyze the running time by counting the number of times each leaf is involved in this charging scheme.  Note that we only call {\sc UpdateCounterSubtree} at $c_1$, which is the lighter subtree.  Clearly, each leaf $u$ belongs to at most $O(\log n)$ light subtrees; hence, it is charged by at most $O(\log n)$ time.  Summing all leaves, we have that the total running time for {\sc UpdateCounterSubtree} is $O(n\log n)$.
\end{proof}

\section{Acknowledgements}

We would like to thank Pawel Gawrychowski and Oren Weiman for pointing out our erroneous claim and also give us reference to Sung's result~\cite{Sung18}.  As mentioned earlier, Jittat Fakcharoenphol who help advising the author on this manuscript would like to take full responsibility for this mistake.
We would like to thank Jittat Fakcharoenphol for suggesting this problem to work on and for his help in editing this manuscript.

\bibliographystyle{unsrt}
\bibliography{main}

\end{document}